\documentclass[11pt]{article}

\usepackage{fullpage,verbatim,amsmath,amssymb,amsthm,float,graphicx,wrapfigure,setspace}

\floatstyle{ruled}
\newfloat{algorithm}{t}{lop}
\floatname{algorithm}{Algorithm}

\newtheorem{lemma}{Lemma}
\newtheorem{theorem}{Theorem}
\newtheorem{definition}{Definition}
\newtheorem{specification}{Specification}

\begin{document}

\title{Introducing Speculation in Self-Stabilization\\ \large{An Application to Mutual Exclusion}}

\author{Swan Dubois\protect\footnote{swan.dubois@epfl.ch}\\LPD, EPFL, Switzerland \and Rachid Guerraoui\protect\footnote{rachid.guerraoui@epfl.ch}\\LPD, EPFL, Switzerland}

\date{}

\maketitle

\begin{abstract}
Self-stabilization ensures that, after any transient fault, the system recovers in a finite time and eventually exhibits. Speculation consists in guaranteeing that the system satisfies its requirements for any execution but exhibits significantly better performances for a subset of executions that are more probable. A speculative protocol is in this sense supposed to be both robust and efficient in practice. 

We introduce the notion of speculative stabilization which we illustrate through the mutual exclusion problem. We then present a novel speculatively stabilizing mutual exclusion protocol. Our protocol is self-stabilizing for any asynchronous execution. We prove that its stabilization time for synchronous executions is $\left\lceil diam(g)/2\right\rceil$ steps (where $diam(g)$ denotes the diameter of the system). 

This complexity result is of independent interest. The celebrated mutual exclusion protocol of Dijkstra stabilizes in $n$ steps (where $n$ is the number of processes) in synchronous executions and the question whether the stabilization time could be strictly smaller than the diameter has been open since then (almost 40 years). We show that this is indeed possible for any underlying topology. We also provide a lower bound proof that shows that our new stabilization time of $\left\lceil diam(g)/2\right\rceil$ steps is optimal for synchronous executions, even if asynchronous stabilization is not required.
\end{abstract}


\section{Introduction}

The speculative approach to distributed computing \cite{L08c, P01j,J03c,GKQV10c,GKL12c} lies on the inherent trade-of between robustness and efficiency. Indeed, we typically require distributed applications to be safe and live under various hostile conditions such as asynchronism, faults, attacks, and contention. This typically leads to high consumption of system resources, \emph{e.g.} time of computation, which is due to the need to perform synchronizations, redundancies or checking. 

The speculative approach assumes that, even if degraded conditions are indeed possible, they are less probable than friendly conditions (for example, synchronous executions without faults). The underlying idea is to simultaneously ensure that the protocol is correct whatever the execution is (even in degraded conditions) but to optimize it for a subset of executions that are the most probable in practice. Even if this idea was applied in various contexts, it has never been applied to distributed systems tolerant to transient faults, \emph{i.e.} self-stabilizing systems \cite{D74j}. In fact, it was not clear whether self-stabilization and speculation could be even combined because of the specific nature of transient faults, for they could corrupt the state of the entire system. The objective of this paper is to explore this avenue. 

Self-stabilization was introduced by Dijkstra \cite{D74j}. Intuitively, a self-stabilizing system ensures that, after the end of any transient fault, the system reaches in a finite time, without any external help, a correct behavior. In other words, a self-stabilizing system repairs itself from any catastrophic state. Since the seminal work of Dijkstra, self-stabilizing protocols were largely studied (see \emph{e.g.} \cite{D00b, T09bc,H02o}). The main objective has been to design self-stabilizing systems tolerating asynchronism while reducing the stabilization time, \emph{i.e.}, the worst time needed by the protocol to recover a correct behavior over all executions of the system.

Our contribution is twofold. First, we define a new variation of self-stabilization in which the main measure of complexity, the stabilization time, is regarded as a function of the adversary and not as a single value. Indeed, we associate to each adversary (known as a \emph{scheduler} or \emph{daemon} in self-stabilization) the worst stabilization time of the protocol over the set of executions captured by this adversary. Then, we define a speculatively stabilizing protocol as a protocol that self-stabilizes under a given adversary but that exhibits a significantly better stabilization time under another (and weaker) adversary. In this way, we ensure that the protocol stabilizes in a large set of executions but guarantees efficiency only on a smaller set (the one we speculate more probable in practice). For the sake of simplicity, we present our notion of speculative stabilization for two adversaries. It could be easily extended to an arbitrary number of adversaries.

Although the idea of optimizing the stabilization time for some subclass of executions is new, some self-stabilizing protocols satisfy (somehow by accident) our definition of speculative stabilization. For example, the Dijkstra's mutual exclusion protocol stabilization time falls to $n$  steps (the number of processes) in synchronous executions. The question whether one could do better has been open since then, \emph{i.e.} during almost 40 years. We close the question in this paper through the second contribution of this paper.

Indeed, we present a novel speculatively stabilizing mutual exclusion protocol. We prove that its stabilization time for synchronous executions is $\left\lceil diam(g)/2\right\rceil$ steps (where $diam(g)$ denotes the diameter of the system), which significantly improves the bound of Dijkstra's protocol. We prove that we cannot improve it. Indeed, we present a lower bound result on the stabilization time of mutual exclusion for synchronous executions. This result is of independent interest since it remains true beyond the scope of speculation and holds even for a protocol that does not need to stabilize in asynchronous executions.

Designing our protocol went through addressing two technical challenges. First, we require the stabilization of a global property (the uniqueness of critical section) in a time strictly smaller than the diameter of the system, which is counter-intuitive (even for synchronous executions). Second, the optimization of the stabilization time for synchronous executions must not prevent the stabilization for asynchronous ones. 

The key to addressing both challenges was a ``reduction'' to clock synchronization: more specifically, leveraging the self-stabilizing asynchronous unison protocol of \cite{BPV04c} within mutual exclusion. We show that it is sufficient to choose correctly the clock size and to grant the access to critical section upon some clock values to ensure $(i)$ the self-stabilization of the protocol for any asynchronous execution as well as $(ii)$ the optimality of its stabilization time for synchronous ones. This reduction was also, we believe, the key to the genericity of our protocol. Unlike Dijkstra's protocol which assumes an underlying ring shaped communication structure, our protocol runs over any communication structure.

We could derive our lower bound result for synchronous executions based on the observation that a process can gather information at most at distance $d$ in $d$ steps whatever protocol it executes. Hence, in the worst case, it is impossible to prevent two processes from simultaneously entering a critical section during the first $\left\lceil diam(g)/2\right\rceil$ steps of all executions with a deterministic protocol.

The rest of this paper is organized as follows. Section \ref{sec:model} introduces the model and the definitions used through the paper. Section \ref{sec:speculativeStabilization} presents our notion of speculative stabilization. Section \ref{sec:mutualExclusion} presents our mutual exclusion protocol. Section \ref{sec:lowerBound} provides our lower bound result. Section \ref{sec:conclusion} ends the paper with some perspectives.

\section{Model, Definitions, and Notations}\label{sec:model}

We consider the classical model of distributed systems introduced by Dijkstra \cite{D74j}. Processes communicate by atomic reading of neighbors' states and the (asynchronous) adversary of the system is captured by an abstraction called \emph{daemon}.

\paragraph{Distributed protocol.} The distributed system consists of a set of processes that form a communication graph. The processes are vertices in this graph and the set of those vertices is denoted by $V$. The edges of this graph are pairs of processes that can communicate with each other. Such pairs are neighbors and the set of edges is denoted by $E$ ($E\subseteq V^2$). Hence, $g=(V,E)$ is the communication graph of the distributed system. Each vertex of $g$ has a set of variables, each of them ranges over a fixed domain of values. A state $\gamma(v)$ of a vertex $v$ is the vector of values of all variables of $v$ at a given time. An assignment of values to all variables of the graph is a configuration. The set of configurations of $g$ is denoted by $\Gamma$. An action $\alpha$ of $g$ transitions the graph from one configuration to another. The set of actions of $g$ is denoted by $A$ ($A=\{(\gamma,\gamma')|\gamma\in\Gamma ,\gamma'\in\Gamma ,\gamma\neq\gamma'\}$). A \emph{distributed protocol} $\pi$ on $g$ is defined as a subset of $A$ that gathers all actions of $g$ allowed by $\pi$. The set of distributed protocols on $g$ is denoted by $\Pi$ ($\Pi=P(A)$ where, for any set $S$, $P(S)$ denotes the powerset of $S$).

\paragraph{Execution.} Given a graph $g$, a distributed protocol $\pi$ on $g$, an \emph{execution} $\sigma$ of $\pi$ on $g$, starting from a given configuration $\gamma_0$, is a maximal sequence of actions of $\pi$ of the following form $\sigma=(\gamma_0,\gamma_1)(\gamma_1,\gamma_2)(\gamma_2,\gamma_3)\ldots$. An execution is \emph{maximal} if it is either infinite or finite but its last configuration is terminal (that is, there exists no actions of $\pi$ starting from this configuration). The set of all executions of $\pi$ on $g$, starting from all configurations of $\Gamma$, is denoted by $\Sigma_\pi$.

\paragraph{Adversary (daemon).} Intuitively, a daemon is a restriction on the executions of distributed protocols to be considered possible. For a distributed protocol $\pi$, at each configuration $\gamma$, a subset of vertices are \emph{enabled}, that is there exists an action of $\pi$ that modifies their state (formally, $\exists \gamma'\in\Gamma,(\gamma,\gamma')\in \pi, \gamma(v)\neq\gamma'(v)$). The daemon then chooses one of the possible action of $\pi$ starting from $\gamma$ (and hence, selects a subset of enabled vertices that are allowed to modify their state during this action). A formal definition follows.

\begin{definition}[Daemon]
Given a graph $g$, a daemon $d$ on $g$ is a function that associates to each distributed protocol $\pi$ on $g$ a subset of executions of $\pi$, that is $d:\pi\in\Pi\longmapsto d(\pi)\in P(\Sigma_\pi)$.
\end{definition}

Given a graph $g$, a daemon $d$ on $g$ and a distributed protocol $\pi$ on $g$, an execution $\sigma$ of $\pi$ ($\sigma\in\Sigma_\pi$) is \emph{allowed} by $d$ if and only if $\sigma\in d(\pi)$. Also, given a graph $g$, a daemon $d$ on $g$ and a distributed protocol $\pi$ on $g$, we say that $\pi$ \emph{runs} on $g$ under $d$ if we consider that the only possible executions of $\pi$ on $g$ are those allowed by $d$.

Some classical examples of daemons follow. The unfair distributed daemon \cite{KY97j} (denoted by $ud$) is the less constrained one because we made no assumption on its choices (any execution of the distributed protocol is allowed). The synchronous daemon \cite{H90j} (denoted by $sd$) is the one that selects all enabled vertices in each configuration. The central daemon \cite{D74j} (denoted by $cd$) selects only one enabled vertex in each configuration.

This way of viewing daemons as a set of possible executions (for a particular graph $g$) drives a natural partial order over the set of daemons. For a particular graph $g$, a daemon $d$ is more powerful than another daemon $d'$ if all executions allowed by $d'$ are also allowed by $d$. Overall, $d$ has more scheduling choices than $d'$. A more precise definition follows.

\begin{definition}[Partial order over daemons]
For a given graph $g$, we define the following partial order $\preccurlyeq$ on $\mathcal{D}$: $\forall (d,d')\in\mathcal{D},d\preccurlyeq d'\Leftrightarrow (\forall \pi\in\Pi,d(\pi)\subseteq d'(\pi))$. If two daemons $d$ and $d'$ satisfy $d\preccurlyeq d'$, we say that $d'$ is more powerful than $d$.
\end{definition}

For example, the unfair distributed daemon is more powerful than any daemon (in particular the synchronous one). Note that some daemons (for example the synchronous and the central ones) are not comparable. For a more detailed discussion about daemons, the reader is referred to \cite{DT11r}.

\paragraph{Further notations.} Given a graph $g$ and a distributed protocol $\pi$ on $g$, we introduce the following set of notations. First, $n$ denotes the number of vertices of the graph whereas $m$ denotes the number of edges ($n=|V|$ and $m=|E|$). The set of neighbors of a vertex $v$ is denoted by $neig(v)$. The distance between two vertices $u$ and $v$ (that is, the length of a shortest path between $u$ and $v$ in $g$) is denoted by $dist(g,u,v)$. The diameter of $g$ (that is, the maximal distance between two vertices of $g$) is denoted by $diam(g)$. For any execution $e=(\gamma_0,\gamma_1)(\gamma_1,\gamma_2)\ldots$, we denote by $e_i$ the prefix of $e$ of length $i$ (that is $e_i=(\gamma_0,\gamma_1)(\gamma_1,\gamma_2)\ldots(\gamma_{i-1},\gamma_i)$).

\paragraph{Guarded representation of distributed protocols.} For the sake of clarity, we do not describe distributed protocols by enumerating all their actions. Instead, we represent distributed protocols using a local description of actions borrowed from \cite{D74j}. Each vertex has a local protocol consisting of a set of guarded rules of the following form: $<\textit{label}>~::~<\textit{guard}>~\longrightarrow~<\textit{action}>$. $<\textit{label}>$ is a name to refer to the rule in the text. $<\textit{guard}>$ is a predicate that involves variables of the vertex and of its neighbors. This predicate is true if and only if the vertex is enabled in the current configuration. We say that a rule is enabled in a configuration when its guard is evaluated to true in this configuration. $<\textit{action}>$ is a set of instructions modifying the state of the vertex. This set of instructions must describe the changes of the vertex state if this latter is activated by the daemon.

\paragraph{Self-stabilization.} Intuitively, to be self-stabilizing \cite{D74j}, a distributed protocol must satisfy the two following properties: $(i)$ \emph{closure}, that is there exists some configuration from which any execution of the distributed protocol satisfies the specification; and $(ii)$ \emph{convergence}, that is starting from any arbitrary configuration, any execution of the distributed protocol reaches in a finite time a configuration that satisfies the closure property.

Self-stabilization induces fault-tolerance since the initial configuration of the system may be arbitrary because of a burst of transient faults. Then, a self-stabilizing distributed protocol ensures that after a finite time (called the convergence or stabilization time), the distributed protocol recovers on his own a correct behavior (by convergence property) and keeps this correct behavior until there is no faults (by closure property). 

\begin{definition}[Self-stabilization \cite{D74j}]
A distributed protocol $\pi$ is self-stabilizing for specification $spec$ under a daemon $d$ if starting from any arbitrary configuration every execution of $d(\pi)$ contains a configuration from which every execution of $d(\pi)$ satisfies $spec$.
\end{definition}

For any self-stabilizing distributed protocol $\pi$ under a daemon $d$ for a specification $spec$, its convergence (or stabilization) time (denoted by $conv\_time(\pi,d)$) is the worst stabilization time (that is, the number of actions required to reach a configuration from which any execution satisfies $spec$) of executions of $\pi$ allowed by $d$. Note that, for any self-stabilizing distributed protocol $\pi$ under a daemon $d$, $\pi$ is self-stabilizing under any daemon $d'$ such that $d'\preccurlyeq d$ and $conv\_time(\pi,d')\leq conv\_time(\pi,d)$.

\section{Speculative Stabilization}\label{sec:speculativeStabilization}

Intuitively, a speculative protocol ensures the correctness in a large set of executions but is optimized for some scenarios that are speculated to be more frequent (maybe at the price of worst performance in less frequent cases). 

Regarding self-stabilization, the most common measure of complexity is the stabilization time. Accordingly, we choose to define a speculatively stabilizing protocol as a self-stabilizing protocol under a given daemon that exhibits a significantly better stabilization time under a weaker daemon (the latter gathers scenarios that are speculated to be more frequent). We can now define our notion of speculative stabilization.

\begin{definition}[Speculative Stabilization]
For two daemons $d$ and $d'$ satisfying $d'\prec d$, a distributed protocol $\pi$ is $(d,d',f)$-speculatively stabilizing for specification $spec$ if: $(i)$ $\pi$ is self-stabilizing for $spec$ under $d$; and $(ii)$ $f$ is a function on $g$ satisfying $\frac{conv\_time(\pi,d)}{conv\_time(\pi,d')}\in\Omega(f)$.
\end{definition}

We restrict ourselves for two daemons here for the sake of clarity. We can easily extend this definition to an arbitrary number of daemons (as long as they are comparable). For instance, we can say that a distributed protocol $\pi$ is $(d,d_1,d_2,f_1,f_2)$-speculatively stabilizing (with $d_1\prec d$ and $d_2\prec d$) if it is both $(d,d_1,f_1)$-speculatively stabilizing and $(d,d_2,f_2)$-speculatively stabilizing.

Still for the sake of simplicity, we say in the following that a distributed protocol $\pi$ is $d$-speculatively stabilizing for specification $spec$ if there exists a daemon $d\neq ud$ such that $\pi$ is $(ud,d,f)$-speculatively stabilizing for specification $spec$ with $f>1$. In other words, a $d$-speculatively stabilizing distributed protocol is self-stabilizing under the unfair distributed daemon (and hence always guarantees convergence) but is optimized for a given subclass of executions described by $d$.

\paragraph{Examples.} Although the idea of speculation approaches in self-stabilization has not been yet precisely defined, there exists some examples of self-stabilizing distributed protocols in the literature that turn out to be speculative. We survey some of them in the following.

The seminal work of Dijkstra \cite{D74j} introduced self-stabilization in the context of mutual exclusion. His celebrated protocol operates only on rings. It is in fact  $(ud,sd,g\mapsto n)$-speculatively stabilizing since it stabilizes upon $\Theta(n^2)$ steps under the unfair distributed daemon and it is easy to see that it needs only $n$ steps to stabilize under the synchronous daemon. The well-known $min+1$ protocol of \cite{HC92j} is $(ud,sd,g\mapsto n^2/diam(g))$-speculatively stabilizing for BFS spanning tree construction. Its stabilization time is in $\Theta(n^2)$ steps under the unfair distributed daemon while it is in $\Theta(diam(g))$ steps under the synchronous daemon. Another example is the self-stabilizing maximal matching protocol of \cite{MMPT09j}. This protocol is $(ud,sd,g\mapsto m/n)$-speculatively stabilizing: its stabilization time is 4n+2m (respectively 2n+1) steps under the unfair distributed (respectively synchronous) daemon. 

\section{A new Mutual Exclusion Protocol}\label{sec:mutualExclusion}

Mutual exclusion was classically adopted as a benchmark in self-stabilization under various settings \cite{D74j,KY02j,DHT04c,CSZ08c,BB11c}. Intuitively, it consists in ensuring that each vertex can enter infinitely often in critical section and there is never two vertices simultaneously in the critical section. Using such a distributed protocol, vertices can for example access shared resources without conflict.

Our contribution in this context is a novel self-stabilizing distributed protocol for mutual exclusion under the unfair distributed daemon that moreover exhibits optimal convergence time under the synchronous daemon. Contrary to the Dijkstra's protocol, our protocol supports any underlying communication structure (we do not assume that the communication graph is reduced to a ring). Thanks to speculation, our protocol is ideal for environment in which we speculate that most of the executions are synchronous.

We adopt the following specification of mutual exclusion. For each vertex $v$, we define a predicate $privileged_v$ (over variables of $v$ and possibly of its neighbors). We say that a vertex $v$ is privileged in a configuration $\gamma$ if and only if $privileged_v=true$ in $\gamma$. If a vertex $v$ is privileged in a configuration $\gamma$ and $v$ is activated during an action $(\gamma,\gamma')$, then $v$ executes its critical section during this action. We can now specify the mutual exclusion problem as follows.

\begin{specification}[Mutual exclusion $spec_{ME}$]
An execution $e$ satisfies $spec_{ME}$ if at most one vertex is privileged in any configuration of $e$ (safety) and any vertex infinitely often executes its critical section in $e$ (liveness).
\end{specification}

The rest of this section is organized as follows. Section \ref{sub:protocol} overviews our protocol. Section \ref{sub:correctness} proves the correctness of our protocol under the unfair distributed daemon. Section \ref{sub:analysis} analyzes its stabilization time under the synchronous and the unfair distributed daemon.

\subsection{Speculatively Stabilizing Mutual Exclusion}\label{sub:protocol}

As we restrict ourselves to deterministic protocols, we know by \cite{BP89j} that, to ensure mutual exclusion, we must assume a system with identities (that is, each vertex has a distinct identifier). Indeed, we know by \cite{BP89j} that the problem does not admit deterministic solution on uniform (\emph{i.e.} without identifiers) rings of composite size. Without loss of generality, we assume that the set of identities (denoted by $ID$) is equals to $\{0,1,\ldots,n-1\}$ (if this assumption is not satisfied, it is easy to define a mapping of identities satisfying it).

Our protocol is based upon an existing self-stabilizing distributed protocol for the asynchronous unison problem \cite{GH90j,CFG92c}. This problem consists in ensuring, under the unfair distributed daemon, some synchronization guarantees on vertices' clocks. More precisely, each vertex has a register $r_v$ that contains a clock value. A clock is a bounded set enhanced with an incrementation function. Intuitively, an asynchronous unison protocol ensures that the difference between neighbors' registers is bounded and that each register is infinitely often incremented.

In the following, we give the definition of this problem and the solution proposed in \cite{BPV04c} from which we derive our mutual exclusion protocol.

\paragraph{Clock.} A bounded clock $\mathcal{X}=(C,\phi)$ is a bounded set $C=cherry(\alpha,K)$ (parametrized with two integers $\alpha\geq 1$ and $K\geq 2$) enhanced with an incrementation function $\phi$ defined as follows.

\begin{wrapfigure}{r}{6cm}
\noindent\begin{center} 
\includegraphics[height=5cm]{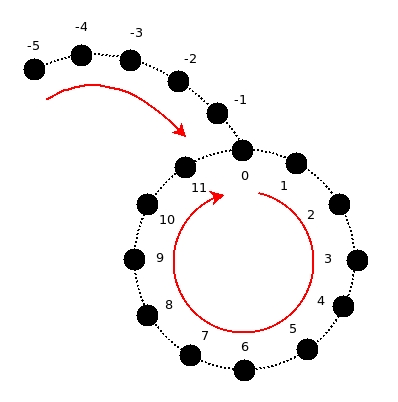}
\end{center}
\caption{A bounded clock $\mathcal{X}=(cherry(\alpha,K),\phi)$ with $\alpha=5$ and $K=12$.}
\label{fig:cherry}
\end{wrapfigure}

Let $c$ be any integer. Denote by $\overline{c}$ the unique element in $[0,\ldots,K-1]$ such that $c=\overline{c}$ mod $K$. We define the distance $d_K(c,c')=min\{\overline{c-c'},\overline{c'-c}\}$ on $[0,\ldots,K-1]$. Two integers $c$ and $c'$ are said to be locally comparable if and only if $d_K(a,b)\leq 1$. We then define the local order relation $\leq_l$ as follows: $c\leq_l c'$ if and only if $0\leq \overline{c'-c}\leq 1$. Let us define $cherry(\alpha,K)=\{-\alpha,\ldots,0,\ldots,K-1\}$. Let $\phi$ be the function defined by:
\[\phi:c\in cherry(\alpha,K)\mapsto\left\{\begin{array}{ll}
(c+1) & \text{if } c<0\\
(c+1) \text{ mod } K & \text{otherwise}
\end{array}\right.\]

The pair $\mathcal{X}=(cherry(\alpha,K),\phi)$ is called a bounded clock of initial value $\alpha$ and of size $K$ (see Figure \ref{fig:cherry}). We say that a clock value $c\in cherry(\alpha,K)$ is incremented when this value is replaced by $\phi(c)$. A reset on $\mathcal{X}$ consists of an operation replacing any value of $cherry(\alpha,K)\setminus\{-\alpha\}$ by $-\alpha$. Let $init_\mathcal{X}=\{-\alpha,\ldots,0\}$ and $stab_\mathcal{X}=\{0,\ldots,K-1\}$ be the set of initial values and correct values respectively. Let us denote $init^*_\mathcal{X}=init_\mathcal{X}\setminus\{0\}$, $stab^*_\mathcal{X}=stab_\mathcal{X}\setminus\{0\}$, and $\leq_{init}$ the usual total order on $init_\mathcal{X}$.

\paragraph{Asynchronous unison.} Given a distributed system in which each vertex $v$ has a register $r_v$ taken a value of a bounded clock $\mathcal{X}=(C,\phi)$ with $C=cherry(\alpha,K)$, we define a legitimate configuration for asynchronous unison as a configuration satisfying: $\forall v\in V,\forall u\in neig(v),(r_v\in stab_\mathcal{X})\wedge(r_u\in stab_\mathcal{X})\wedge(d_K(r_v,r_u)\leq 1)$. In other words, a legitimate configuration is a configuration in which each clock value is a correct one and the drift between neighbors' registers is bounded by $1$. We denote by $\Gamma_1$ the set of legitimate configurations for asynchronous unison. Note that we have, for any configuration of $\Gamma_1$ and any pair of vertices, $(u,v)$, $d_K(r_u,r_v)\leq diam(g)$ by definition. We can now specify the problem.

\begin{specification}[Asynchronous unison $spec_{AU}$]
An execution $e$ satisfies $spec_{AU}$ if every configuration of $e$ belongs to $\Gamma_1$ (safety) and the clock value of each vertex is infinitely often incremented in $e$ (liveness).
\end{specification}

In \cite{BPV04c}, the authors propose a self-stabilizing asynchronous unison distributed protocol in any anonymous distributed system under the unfair distributed daemon. The main idea of this protocol is to reset the clock value of each vertex that detects any local safety violation (that is, whenever some neighbor that has a not locally comparable clock value). Otherwise, a vertex is allowed to increment its clock (of initial or of correct value) only if this latter has locally the smallest value. The choice of parameters $\alpha$ and $K$ are crucial. In particular, to make the protocol self-stabilizing for any anonymous communication graph $g$ under the unfair distributed daemon, the parameters must satisfy $\alpha\geq hole(g)-2$ and $K>cyclo(g)$, where $hole(g)$ and $cyclo(g)$ are two constants related to the topology of $g$. Namely, $hole(g)$ is the length of a longest hole in $g$ (\emph{i.e.} the longest chordless cycle), if $g$ contains a cycle, $2$ otherwise. $cyclo(g)$ is the cyclomatic characteristic of $g$ (\emph{i.e.} the length of the maximal cycle of the shortest maximal cycle basis of $g$), if $g$ contains a cycle, $2$ otherwise. Actually, \cite{BPV04c} shows that taking $\alpha\geq hole(g)-2$ ensures that the protocol recovers in finite time a configuration in $\Gamma_1$. Then, taking $K>cyclo(g)$ ensures that each vertex increments its local clock infinitely often. Note that, by definition, $hole(g)$ and $cyclo(g)$ are bounded by $n$.

\paragraph{The mutual exclusion protocol.} The main idea behind our protocol is to execute the asynchronous unison of \cite{BPV04c}, presented earlier, with a particular bounded clock and then to grant the privilege to a vertex only when its clock reaches some value. The clock size must be sufficiently large to ensure that at most one vertex is privileged in any configuration of $\Gamma_1$. If the definition of the predicate $privileged$ guarantees this property, then the correctness of our mutual exclusion protocol follows from the one of the underlying asynchronous unison.

More specifically, we choose a bounded clock $\mathcal{X}=(cherry(\alpha,K),\phi)$ with $\alpha=n$ and $K=(2.n-1)(diam(g)+1)+2$ and we define $privileged_v \equiv (r_v=2.n+2.diam(g).id_v)$. In particular, note that we have : $privileged_{v_0} \equiv (r_{v_0}=2.n)$ and $privileged_{v_{n-1}} \equiv (r_{v_{n-1}}=(2.n-2)(diam(g)+1)+2)$.

Our distributed protocol, called $\mathcal{SSME}$ (for $\mathcal{S}$peculatively $\mathcal{S}$tabilizing $\mathcal{M}$utual $\mathcal{E}$xclusion), is described in Algorithm 1. Note that this protocol is identical to the one of \cite{BPV04c} except for the size of the clock and the definition of the predicate $privileged$ (that does not interfere with the protocol).

We prove in the following that this protocol is self-stabilizing for $spec_{ME}$ under the unfair distributed daemon and exhibits the optimal convergence time under the synchronous one. In other words, we will prove that this protocol is $sd$-speculatively stabilizing for $spec_{ME}$. 

\begin{algorithm}
\caption{$\mathcal{SSME}$: Mutual exclusion protocol for vertex $v$.}
\textbf{Constants:}\\
$\begin{array}{lll}
id_v\in ID & : & \text{identity of } v\\
n\in\mathbb{N} & : & \text{number of vertices of the communication graph}\\
diam(g)\in\mathbb{N} & : & \text{diameter of the communication graph}\\
\mathcal{X}=(cherry(n,(2.n-1)(diam(g)+1)+2),\phi) & : & \text{clock of $v$}
\end{array}$\\
\textbf{Variable:}\\
$\begin{array}{lll}
r_v\in \mathcal{X} & : & \text{register of } v
\end{array}$\\
\textbf{Predicates:}\\
$\begin{array}{lll}
privileged_v & \equiv & (r_v=2.n+2.diam(g).id_v)\\
correct_v(u) & \equiv & (r_v\in stab_\mathcal{X})\wedge(r_u\in stab_\mathcal{X})\wedge(d_K(r_v,r_u)\leq 1)\\
allCorrect_v & \equiv & \forall u\in neig(v), correct_v(u)\\
normalStep_v & \equiv & allCorrect_v\wedge(\forall u\in neig(v),r_v\leq_l r_u)\\
convergeStep_v & \equiv & r_v\in init^*_\mathcal{X}\wedge\forall u\in neig(v),(r_u\in init_\mathcal{X}\wedge r_v\leq_{init} r_u)\\
resetInit_v & \equiv & \neg allCorrect_v\wedge(r_v\notin init_\mathcal{X})
\end{array}$\\
\textbf{Rules:}\\
$\begin{array}{lllll}
NA & :: & normalStep_v & \longrightarrow & r_v:=\phi(r_v)\\
CA & :: & convergeStep_v & \longrightarrow & r_v:=\phi(r_v)\\
RA & :: & resetInit_v & \longrightarrow & r_v:=-n
\end{array}$
\end{algorithm}

\subsection{Correctness}\label{sub:correctness}

We prove here the self-stabilization of $\mathcal{SSME}$ under the unfair distributed daemon.

\begin{theorem}\label{th:correctness}
$\mathcal{SSME}$ is a self-stabilizing distributed protocol for $spec_{ME}$ under $ufd$.
\end{theorem}

\begin{proof}
As we choose $\alpha=n\geq hole(g)-2$ and $K=(2.n-1)(diam(g)+1)+2>n\geq cyclo(g)$, the main result of \cite{BPV04c} allows us to deduce that $\mathcal{SSME}$ is a self-stabilizing distributed protocol for $spec_{AU}$ under $ufd$ (recall that the predicate $privileged$ does not interfere with the protocol). By definition, this implies that there exists, for any execution $e$ of $\mathcal{SSME}$ under $ufd$, a suffix $e'$ reached in a finite time that satisfies $spec_{AU}$.

Let $\gamma$ be a configuration of $e'$ such a vertex $v$ is privileged in $\gamma$. Then, by definition, we have $r_v=2.n+2.diam(g).id_v$. As $\gamma$ belongs to $e'$, we can deduce that $\gamma\in\Gamma_1$. Hence, for any vertex $u\in V\setminus\{v\}$, we have $d_K(r_u,r_v)\leq diam(g)$. Then, by definition of the predicate $prvileged$, no other vertex than $v$ can be privileged in $\gamma$. We can deduce that the safety of $spec_{ME}$ is satisfied on $e'$. The liveness of $spec_{ME}$ on $e'$ follows from the one of $spec_{AU}$ and from the definition of the predicate $privileged$.

Hence, for any execution of $\mathcal{SSME}$ under $ufd$, there exists a suffix reached in a finite time that satisfies $spec_{ME}$, that proves the theorem.
\end{proof}

\subsection{Time Complexities}\label{sub:analysis}

This section analyses the time complexity of our self-stabilizing mutual exclusion protocol. In particular, we provide an upper bound of its stabilization time under the synchronous daemon (see Theorem \ref{th:stabsd}) and  under the unfair distributed daemon (see Theorem \ref{th:stabufd}).
 
\paragraph{Synchronous daemon.} We first focus on the stabilization time of $\mathcal{SSME}$ under the synchronous daemon. We need to introduce some notations and definitions.

From now, $e=(\gamma_0,\gamma_1)(\gamma_1,\gamma_2)\ldots$ denotes a synchronous execution of $\mathcal{SSME}$ starting from an arbitrary configuration $\gamma_0$. For a configuration $\gamma_i$ and a vertex $v$, $r^i_v$ denotes the value of $r_v$ in $\gamma_i$.

\begin{definition}[Island]
In a configuration $\gamma_i$, an island $I$ is a maximal (w.r.t. inclusion) set of vertices such that $I\subsetneq V$ and $\forall (u,v)\in I, u\in neig(v)\Rightarrow correct_v(u)$. A zero-island is an island such that $\exists v\in I, r^i_v=0$. A non-zero-island is an island such that $\forall v\in I, r^i_v\neq 0$.
\end{definition}

Note that any vertex $v$ that satisfies $r_v\in stab_\mathcal{X}$ in a configuration $\gamma\notin \Gamma_1$ belongs by definition to an island (either a zero-island or a non-zero-island) in $\gamma$.

\begin{definition}[Border and depth of an island]
In a configuration $\gamma_i$ that contains an island $I\neq\emptyset$, the border of $I$ (denoted by $border(I)$) is defined by $border(I)=\{v\in I|\exists u\in V\setminus I, u\in neig(v)\}$ and the depth of $I$ (denoted by $depth(I)$) is defined by $depth(I)=max\{min\{dist(g,v,u)|u\in border(I)\}|v\in I\}$.
\end{definition}

Then, we have to prove a set of preliminaries lemmas before stating our main theorem.

\begin{lemma}\label{lem:CARA}
If a vertex $v$ is privileged in a configuration $\gamma_i$ (with $0\leq i<diam(g)$), then $v$ cannot execute rules $CA$ and $RA$ in $e_i$.
\end{lemma}

\begin{proof}
As the result is obvious for $i=0$, let $\gamma_i$ (with $0<i<diam(g)$) be a configuration such that a vertex $v$ is privileged in $\gamma_i$. Then, we have by definition that $r^i_v=2.n+2.diam(g).id_v$.

By contradiction, assume that $v$ executes at least once rule $CA$ or $RA$ in $e_i$. Let $j$ be the biggest integer such that $v$ executes rule $CA$ or $RA$ during action $(\gamma_j,\gamma_{j+1})$ with $j<i$. 

Assume that $v$ executes rule $RA$ during $(\gamma_j,\gamma_{j+1})$. Then, we have $r^{j+1}_v=-n$. From this point, only rule $CA$ may be enabled at $v$ but $v$ does not execute it by construction of $j$. Then, we can deduce that $r^i_v=-n$ that is contradictory.

Hence, we know that $v$ executes rule $CA$ during $(\gamma_j,\gamma_{j+1})$. Consequently, we have $r^{j+1}_v\in init_\mathcal{X}$ by construction of the rule. As $v$ can only execute rule $NA$ between $\gamma_{j+1}$ and $\gamma_i$ by construction of $j$, we can deduce that $r^i_v\in init_\mathcal{X}\cup\{0,\ldots,0+i-(j+1)\}$. As $ 0+i-(j+1)<diam(g)$, this contradiction proves the result.
\end{proof}

\begin{lemma}\label{lem:zeroIsland}
If a vertex $v$ is privileged in a configuration $\gamma_i$ (with $0\leq i<diam(g)$), then $v$ cannot belong to a zero-island in any configuration of $e_i$.
\end{lemma}

\begin{proof}
Let $\gamma_i$ (with $0\leq i<diam(g)$) be a configuration such that a vertex $v$ is privileged in $\gamma_i$. Then, we have by definition that $r^i_v=2.n+2.diam(g).id_v$. 

By contradiction, assume that there exists some configurations of $e_i$ such that $v$ belongs to a zero-island. Let $j$ be the biggest integer such that $v$ belongs to a zero-island $I$ in $\gamma_j$ with $j\leq i$.

By definition of a zero-island, we know that there exists a vertex $u$ in $I$ such that $r^j_u=0$. As $dist(g,u,v)\leq diam(g)$ and $u$ and $v$ belongs to the same island in $\gamma_j$, we have $d_K(r^j_u,r^j_v)\leq diam(g)$. By construction of the clock, we have so $r^j_v\in\{(2.n-2)(diam(g)+1)+3,\ldots,0,\ldots,diam(g)\}$.

By Lemma \ref{lem:CARA}, we know that $v$ may execute only rule $NA$ between $\gamma_j$ and $\gamma_j$. Then, we have $r^i_v\in\{(2.n-2)(diam(g)+1)+3,\ldots,0,\ldots,diam(g)+(i-j)\}$. As $diam(g)+(i-j)<2.diam(g)$, $v$ cannot be privileged in $\gamma_i$ (whatever is its identity). This contradiction proves the result.
\end{proof}

\begin{lemma}\label{lem:nonZeroIsland}
If a vertex $v$ belongs to a non-zero-island of depth $k\geq 0$ in a configuration $\gamma_i$ (with $0<i<diam(g)$), then $v$ belongs either to a non-zero-island of depth greater or equals to $k+1$ or to a zero-island in $\gamma_{i-1}$.
\end{lemma}

\begin{proof}
Let $\gamma_i$ (with $0<i<diam(g)$) be a configuration such that a vertex $v$ belongs to a non-zero-island $I$ of depth $k\geq 0$ in $\gamma_i$.

Assume that $v$ does not belongs to any island in $\gamma_{i-1}$. In other words, we have $r^{i-1}_v\in init^*_\mathcal{X}$. Consequently, $v$ may only execute rule $CA$ during action $(\gamma_{i-1},\gamma_i)$ and we have $r^{i}_v\in init_\mathcal{X}$. This means that $v$ either belongs to a zero-island or does not belong to any island in $\gamma_i$. This contradiction shows us that $v$ belongs to an island in $\gamma_{i-1}$.

If $v$ belongs to a zero-island in $\gamma_{i-1}$, we have the result. Otherwise, assume by contradiction that $v$ belongs to a non-zero island $I'$ such that $depth(I')\leq k$ in $\gamma_{i-1}$. By definition of a non-zero-island, all vertices of $border(I')$ are enabled by rule $RA$ in $\gamma_{i-1}$. As we consider a synchronous execution, we obtain that $I$ (the non-zero-island that contains $v$ in $\gamma_i$) satisfies $depth(I)<k$. This contradiction shows the lemma.
\end{proof}

\begin{lemma}\label{lem:gamma0}
If $\gamma_0\notin\Gamma_1$, then any vertex $v$ satisfies $r^{diam(g)}_v\in init_\mathcal{X}\cup\{(2.n-2)(diam(g)+1)+3,\ldots,0,\ldots,2.diam(g)-1\}$.
\end{lemma}

\begin{proof}
Assume that $\gamma_0\notin\Gamma_1$. Then, by definition of $\Gamma_1$ and by the construction of the protocol, we know that there exists a set $\emptyset\neq V'\subseteq V$ such that vertices of $V'$ are enabled by rule $RA$ in $\gamma_0$. Let $v$ be an arbitrary vertex of $V$.

If $v$ executes at least once the rule $RA$ during $e_{diam(g)}$, let $i$ be the biggest integer such that $v$ executes rule $RA$ during $(\gamma_i,\gamma_{i+1})$ with $i<diam(g)$. Then, we have $r^{i+1}_v=-n$. As $diam(g)-(i+1)<n$, we can deduce that $v$ may execute only rule $CA$ between $\gamma_i$ and $\gamma_{diam(g)}$. Consequently, we have $r^{diam(g)}_v\in init_\mathcal{X}$.

If $v$ executes at least once the rule $CA$ but never executes rule $RA$ during $e_{diam(g)}$, let $i$ be the biggest integer such that $v$ executes rule $CA$ during $(\gamma_i,\gamma_{i+1})$ with $i<diam(g)$. Then, we have $r^{i+1}_v\in init_\mathcal{X}$. By construction of $i$, we can deduce that $v$ may execute only rule $NA$ between $\gamma_i$ and $\gamma_{diam(g)}$. As $diam(g)-(i+1)<diam(g)$, we have $r^{diam(g)}_v\in init_\mathcal{X}\cup\{0,\ldots,diam(g)-1\}$.

Otherwise ($v$ executes only rule $NA$ during $e_{diam(g)}$), let $i$ be the integer defined by $i=min\{dist(g,v,v')|v'\in V'\}$. Note that $0<i\leq diam(g)$ by construction (recall that $v\notin V'$). We can deduce that $v$ belongs to a zero-island in $\gamma_i$ (otherwise, $v$ executes rule $RA$ or $CA$ during $(\gamma_i,\gamma_{i+1})$). By definition of a zero-island, we have then $r^{i}_v\in\{(2.n-2)(diam(g)+1)+3,\ldots,0,\ldots diam(g)\}$. As $v$ may execute only rule $NA$ between $\gamma_i$ and $\gamma_{diam(g)}$ and $diam(g)-i<diam(g)$, we can deduce that $r^{diam(g)}_v\in \{(2.n-2)(diam(g)+1)+3,\ldots,0,\ldots,2.diam(g)-1\}$.
\end{proof}

\begin{theorem}\label{th:stabsd}
$conv\_time(\mathcal{SSME},sd)\leq\left\lceil\frac{diam(g)}{2}\right\rceil$
\end{theorem}

\begin{proof}
By contradiction, assume that $conv\_time(\mathcal{SSME},sd)>\left\lceil\frac{diam(g)}{2}\right\rceil$. This means that there exists a configuration $\gamma_0$ such that the synchronous execution $e=(\gamma_0,\gamma_1)(\gamma_1,\gamma_2)\ldots$ of $\mathcal{SSME}$ satisfies: there exists an integer $i\geq\left\lceil\frac{diam(g)}{2}\right\rceil$ and two vertices $u$ and $v$ such that $u$ and $v$ are simultaneously privileged in $\gamma_i$. Let us study the following cases (note that they are exhaustive):

\begin{description}
\item[Case 1:] $\left\lceil\frac{diam(g)}{2}\right\rceil\leq i<diam(g)$

By Lemma \ref{lem:CARA}, we know that $u$ may execute only rule $NA$ in $e_i$. This implies that $\forall j\leq i, r^j_u\in stab_\mathcal{X}$ and then $d_K(r^i_u,r^0_u)\leq i$. By the same way, we can prove that $d_K(r^i_v,r^0_v)\leq i$.

If $u$ is privileged in $\gamma_i$, this means that $r^i_u\in stab_\mathcal{X}$ and $d_K(r^i_u,0)>diam(g)$. As $u$ and $v$ are simultaneously privileged in $\gamma_i$, we have by definition that $d_K(r^i_u,r^i_v)>diam(g)$. This implies that $\gamma_i\notin \Gamma_1$ and that $u$ belongs to a non-zero-island $I$ such that $depth(I)\geq 1$ in $\gamma_i$. By recursive application of Lemmas \ref{lem:zeroIsland} and \ref{lem:nonZeroIsland}, we deduce that $u$ belongs to a non-zero-island $I'$ such that $depth(I')\geq i+1\geq \left\lceil\frac{diam(g)}{2}\right\rceil+1$ in $\gamma_0$. The same property holds for $v$. As $dist(g,u,v)\leq diam(g)$, we can deduce that $u$ and $v$ belongs to the same non-zero-island in $\gamma_0$, that allows us to state $d_K(r^0_u,r^0_v)\leq diam(g)$. 

Without loss of generality, assume that $id_u<id_v$. Let us now distinguish the following cases:

If $id_v-id_u\geq 2$, as $u$ and $v$ are simultaneously privileged in $\gamma_i$, we have $d_K(r^i_u,r^i_v)\geq 2.n+diam(g)+1$ (if $id_u=n-1$ and $id_v=0$) or $d_K(r^i_u,r^i_v)\geq 4.diam(g)$ (otherwise). Note that in both cases, we have $d_K(r^i_u,r^i_v)\geq 3.diam(g)$. Recall that $d_K$ is a distance. In particular, it must satisfy the triangular inequality. Then, we have $d_K(r^i_u,r^i_v)\leq d_K(r^i_u,r^0_u)+d_K(r^0_u,r^0_v)+d_K(r^0_u,r^i_v)$. By previous result, we obtain that $d_K(r^i_u,r^i_v)\leq diam(g)+2.i<3.diam(g)$, that is contradictory.

If $id_v-id_u=1$, by construction of $\gamma_i$, we have $r^i_u=2.n+2.diam(g).id_u>0$ and $r^i_v=2.n+2.diam(g).$ $(id_u+1)$. Then, we obtain $r^i_v-r^i_u=2.diam(g)$. Hence, we have $0<r^0_u\leq r^i_u<r^0_v\leq r^i_v$. Then, we can deduce from $r^i_v-r^i_u=2.diam(g)$ and $r^i_u-r^o_u\geq 0$ that $r^i_v-r^0_u\geq 2.diam(g)$. On the other hand, previous results show us that $r^0_v-r^0_u\leq diam(g)$ and $r^i_v-r^0_v<diam(g)$. It follows $r^i_v-r^0_u<2.diam(g)$, that is contradictory.

\item[Case 2:] $diam(g)\leq i<2.n+diam(g)$

As $u$ and $v$ are simultaneously privileged in $\gamma_i$, we have by definition that $d_K(r^i_u,r^i_v)>diam(g)$. This implies that $\gamma_i\notin \Gamma_1$ and then $\gamma_0\notin\Gamma_1$ (otherwise, we obtain a contradiction with the closure of $spec_{AU}$).

By Lemma \ref{lem:gamma0}, for any vertex $w$, $r^{diam(g)}_w\in init_\mathcal{X}\cup\{(2.n-2)(diam(g)+1)+3,\ldots,0,\ldots,2.diam(g)-1\}$. As $w$ may execute at most $i-diam(g)<2.n$ actions between $\gamma_{diam(g)}$ and $\gamma_i$, we can deduce that $r^{i}_w\in init_\mathcal{X}\cup\{(2.n-2)(diam(g)+1)+3,\ldots,0,\ldots,2.n+2.diam(g)-1\}$ for any vertex $w$.

By construction of the clock and the definition of the predicate $privileged$, we can conclude that there is at most one privileged vertex (the one with identity $0$) in $\gamma_i$, that is contradictory.

\item[Case 3:] $i\geq 2.n+diam(g)$

By \cite{BPV08j}, we know that $\mathcal{SSME}$ stabilizes to $spec_{AU}$ in at most $\alpha+lcp(g)+diam(g)$ steps under the synchronous daemon where $lcp(g)$ denotes the length of the longest elementary chordless path of $g$. As we have $\alpha=n$ by construction and $lcp(g)\leq n$ by definition, we can deduce that $\mathcal{SSME}$ stabilizes to $spec_{AU}$ in at most $2.n+diam(g)$ steps under the synchronous daemon.

In particular, this implies that $\gamma_i\in\Gamma_1$. Then, using proof of Theorem \ref{th:correctness}, we obtain a contradiction with the fact that $u$ and $v$ are simultaneously privileged in $\gamma_i$.
\end{description}

We thus obtain that  $conv\_time(\mathcal{SSME},sd)\leq\left\lceil\frac{diam(g)}{2}\right\rceil$.
\end{proof}

\paragraph{Unfair distributed daemon.} We now interested in the stabilization time of our mutual exclusion protocol under the unfair distributed daemon. Using a previous result from \cite{DP12c}, we have the following upper bound:

\begin{theorem}\label{th:stabufd}
$conv\_time(\mathcal{SSME},ufd)\in O(diam(g).n^3)$
\end{theorem}

\begin{proof}
Remind that the stabilization time of $\mathcal{SSME}$ for $spec_{AU}$ is an upper bound for the one for $spec_{ME}$ whatever the daemon is. The step complexity of this protocol is tricky to exactly compute. As the best of our knowledge, \cite{DP12c} provides the best known upper bound on this step complexity.

The main result of \cite{DP12c} is to prove that $\mathcal{SSME}$ stabilizes in at most $2.diam(g).n^3+(\alpha+1).n^2+(\alpha-2.diam(g)).n$ steps under $ufd$. Since we chose $\alpha=n$, we have the result.
\end{proof}

\section{Synchronous Lower Bound}\label{sec:lowerBound}

We prove here a lower bound on the stabilization time of mutual exclusion under a synchronous daemon, showing hereby that our speculatively stabilizing protocol presented in Section \ref{sub:protocol} is in this sense optimal. We introduce some definitions and a lemma.
 
\begin{definition}[Local state]
Given a configuration $\gamma$, a vertex $v$ and an integer $0\leq k\leq diam(g)$, the $k$-local state of $v$ in $\gamma$ (denoted by $\gamma_{v,k}$) is the configuration of the communication subgraph $g'=(V',E')$ induced by $V'=\{v'\in V|dist(g,v,v')\leq k\}$ defined by $\forall v'\in V', \gamma_{v,k}(v')=\gamma(v')$.
\end{definition}

Note that $\gamma_{v,0}=\gamma(v)$ by definition.

\begin{definition}[Restriction of an execution]
Given an execution $e=(\gamma_0,\gamma_1)(\gamma_1,\gamma_2)\ldots$ and a vertex $v$, the restriction of $e$ to $v$ (denoted by $e_v$) is defined by $e_v=(\gamma_0(v),\gamma_1(v))(\gamma_1(v),\gamma_2(v))\ldots$.
\end{definition}

\begin{lemma}\label{lem:lowerBound}
For any self-stabilizing distributed protocol $\pi$ for $spec_{ME}$ under the synchronous daemon and any pair of configuration $(\gamma,\gamma')$ such that there exists a vertex $v$ and an integer $1\leq k\leq diam(g)$ satisfying $\gamma_{v,k}=\gamma'_{v,k}$, the restrictions to $v$ of the prefixes of length $k$ of executions of $\pi$ starting respectively from $\gamma$ and $\gamma'$ are equals.
\end{lemma}

\begin{proof}
Let $\pi$ be a self-stabilizing distributed protocol for $spec_{ME}$ under the synchronous daemon and $(\gamma,\gamma')$ two configurations such that there exists a vertex $v$ and an integer $1\leq k\leq diam(g)$ satisfying $\gamma_{v,k}=\gamma'_{v,k}$. We denote by $e=(\gamma,\gamma_1)(\gamma_1,\gamma_2)\ldots$ (respectively $e'=(\gamma',\gamma'_1)(\gamma'_1,\gamma'_2)\ldots$) the synchronous execution of $\pi$ starting from $\gamma$ (respectively $\gamma'$). We are going to prove the lemma by induction on $k$.

For $k=1$, we have $\gamma_{v,1}=\gamma'_{v,1}$, that is the state of $v$ and of its neighbors are identical in $\gamma$ and $\gamma'$. As the daemon is synchronous, we have $(e_1)_v=(e'_1)_v$, that implies the result.

For $k>1$, assume that the lemma is true for $k-1$. The induction assumption and the synchrony of the daemon allows us to deduce that $(e_{k-1})_v=(e'_{k-1})_v$ and $\forall u\in neig(v),(e_{k-1})_u=(e'_{k-1})_u$. Hence, we have $(\gamma_{k-1})_{v,1}=(\gamma'_{k-1})_{v,1}$. Then, by the same argument than in the case $k=1$, we deduce that $(\gamma_{k})_{v,0}=(\gamma'_{k})_{v,0}$, that implies the result.
\end{proof}

\begin{theorem}\label{th:lowerBound}
Any self-stabilizing distributed protocol $\pi$ for $spec_{ME}$ satisfies $conv\_time(\pi,sd)\geq\left\lceil\frac{diam(g)}{2}\right\rceil$.
\end{theorem}

\begin{proof}
By contradiction, assume that there exists a self-stabilizing distributed protocol $\pi$ for $spec_{ME}$ such that $conv\_time(\pi,sd)<\left\lceil\frac{diam(g)}{2}\right\rceil$. For the sake of notation, let us denote $t=conv\_time(\pi,sd)$.

Given an arbitrary communication graph $g$, choose two vertices $u$ and $v$ such that $dist(g,u,v)=diam(g)$ and an arbitrary configuration $\gamma_0$. Denote by $e=(\gamma_0,\gamma_1)(\gamma_1,\gamma_2)\ldots$ the synchronous execution of $\pi$ starting from $\gamma_0$.

By definition, $e$ contains an infinite suffix in which $u$ (respectively $v$) executes infinitely often its critical section. Hence, there exists a configuration $\gamma_i$ (respectively $\gamma_j$) such that $u$ (respectively $v$) is privileged in $\gamma_i$ (respectively $\gamma_j$) and $i>t$ (respectively $j>t$).

As $t<\left\lceil\frac{diam(g)}{2}\right\rceil$ and $dist(g,u,v)=diam(g)$, there exists at least one configuration $\gamma'_0$ such that $(\gamma'_0)_{u,t}=(\gamma_{i-t})_{u,t}$ and $(\gamma'_0)_{v,t}=(\gamma_{j-t})_{v,t}$. Let $e'=(\gamma'_0,\gamma'_1)(\gamma'_1,\gamma'_2)\ldots$ be the synchronous execution of $\pi$ starting from $\gamma'_0$.

By Lemma \ref{lem:lowerBound}, we can deduce that the restriction to $u$ of the prefix of length $t$ of $e'$ is the same as the one of the suffix of $e$ starting from $\gamma_{i-t}$. In particular, $u$ is privileged in $\gamma'_t$. By the same way, we know that $v$ is privileged in $\gamma'_t$. This contradiction leads to the result.
\end{proof}

\section{Conclusion}\label{sec:conclusion}

This paper studies for the first time the notion of speculation in self-stabilization. As the main measure in this context is the stabilization time, we naturally consider that a speculatively stabilizing protocol is a self-stabilizing protocol for a given adversary that exhibits moreover a better stabilization time under another (and weaker) adversary. This weaker adversary captures a subset of most probable executions for which the protocol is optimized.

To illustrate this approach, we consider the seminal problem of Dijkstra on self-stabilization: mutual exclusion. We provide a new self-stabilizing mutual exclusion protocol. We prove then that this protocol has an optimal stabilization time in synchronous executions.

Our paper opens a new path of research in self-stabilization by considering the stabilization time of a protocol as a function of the adversary and not as a single value. As a continuation, one could naturally apply our new notion of speculative stabilization to other classical problems of distributed computing and provide speculative protocols for other adversaries than the synchronous one. It may also be interesting to study a composition tool that automatically ensures speculative stabilization.

\bibliographystyle{plain}
\bibliography{biblio}

\begin{thebibliography}{10}

\bibitem{BB11c}
Joffroy Beauquier and Janna Burman.
\newblock Self-stabilizing mutual exclusion and group mutual exclusion for
  population protocols with covering.
\newblock In {\em OPODIS}, pages 235--250, 2011.

\bibitem{BPV04c}
Christian Boulinier, Franck Petit, and Vincent Villain.
\newblock When graph theory helps self-stabilization.
\newblock In {\em PODC}, pages 150--159, 2004.

\bibitem{BPV08j}
Christian Boulinier, Franck Petit, and Vincent Villain.
\newblock Synchronous vs. asynchronous unison.
\newblock {\em Algorithmica}, 51(1):61--80, 2008.

\bibitem{BP89j}
James~E. Burns and Jan~K. Pachl.
\newblock Uniform self-stabilizing rings.
\newblock {\em ACM Trans. Program. Lang. Syst.}, 11(2):330--344, 1989.

\bibitem{CSZ08c}
Viacheslav Chernoy, Mordechai Shalom, and Shmuel Zaks.
\newblock A self-stabilizing algorithm with tight bounds for mutual exclusion
  on a ring.
\newblock In {\em DISC}, pages 63--77, 2008.

\bibitem{CFG92c}
Jean-Michel Couvreur, Nissim Francez, and Mohamed~G. Gouda.
\newblock Asynchronous unison.
\newblock In {\em ICDCS}, pages 486--493, 1992.

\bibitem{DP12c}
St{\'e}phane Devismes and Franck Petit.
\newblock On efficiency of unison.
\newblock In {\em TADDS}, pages 20--25, 2012.

\bibitem{D74j}
Edsger~W. Dijkstra.
\newblock Self-stabilizing systems in spite of distributed control.
\newblock {\em Communication of ACM}, 17(11):643--644, 1974.

\bibitem{D00b}
Shlomi Dolev.
\newblock {\em Self-stabilization}.
\newblock MIT Press, 2000.

\bibitem{DT11r}
Swan Dubois and S{\'e}bastien Tixeuil.
\newblock A taxonomy of daemons in self-stabilization.
\newblock {\em CoRR}, abs/1110.0334, 2011.

\bibitem{DHT04c}
Philippe Duchon, Nicolas Hanusse, and S{\'e}bastien Tixeuil.
\newblock Optimal randomized self-stabilizing mutual exclusion on synchronous
  rings.
\newblock In {\em DISC}, pages 216--229, 2004.

\bibitem{GH90j}
Mohamed~G. Gouda and Ted Herman.
\newblock Stabilizing unison.
\newblock {\em Information Processing Letters}, 35(4):171--175, 1990.

\bibitem{GKQV10c}
Rachid Guerraoui, Nikola Knezevic, Vivien Qu{\'e}ma, and Marko Vukolic.
\newblock The next 700 bft protocols.
\newblock In {\em EuroSys}, pages 363--376, 2010.

\bibitem{GKL12c}
Rachid Guerraoui, Viktor Kuncak, and Giuliano Losa.
\newblock Speculative linearizability.
\newblock In {\em PLDI}, pages 55--66, 2012.

\bibitem{H90j}
Ted Herman.
\newblock Probabilistic self-stabilization.
\newblock {\em Information Processing Letters}, 35(2):63--67, 1990.

\bibitem{H02o}
Ted Herman.
\newblock A comprehensive bibliography on self-stabilization.
\newblock http://www.cs.uiowa.edu/ftp/selfstab/bibliography/, 2002.

\bibitem{HC92j}
Shing-Tsaan Huang and Nian-Shing Chen.
\newblock A self-stabilizing algorithm for constructing breadth-first trees.
\newblock {\em Information Processing Letters}, 41(2):109--117, 1992.

\bibitem{J03c}
Prasad Jayanti.
\newblock Adaptive and efficient abortable mutual exclusion.
\newblock In {\em PODC}, pages 295--304, 2003.

\bibitem{KY97j}
Hirotsugu Kakugawa and Masafumi Yamashita.
\newblock Uniform and self-stabilizing token rings allowing unfair daemon.
\newblock {\em IEEE Transactions on Parallel and Distributed Systems},
  8(2):154--162, 1997.

\bibitem{KY02j}
Hirotsugu Kakugawa and Masafumi Yamashita.
\newblock Uniform and self-stabilizing fair mutual exclusion on unidirectional
  rings under unfair distributed daemon.
\newblock {\em J. Parallel Distrib. Comput.}, 62(5):885--898, 2002.

\bibitem{L08c}
Butler~W. Lampson.
\newblock Lazy and speculative execution in computer systems.
\newblock In {\em ICFP}, pages 1--2, 2008.

\bibitem{MMPT09j}
Fredrik Manne, Morten Mjelde, Laurence Pilard, and S{\'e}bastien Tixeuil.
\newblock A new self-stabilizing maximal matching algorithm.
\newblock {\em Theoretical Computer Science}, 410(14):1336--1345, 2009.

\bibitem{P01j}
Fernando Pedone.
\newblock Boosting system performance with optimistic distributed protocols.
\newblock {\em IEEE Computer}, 34(12):80--86, 2001.

\bibitem{T09bc}
S\'{e}bastien Tixeuil.
\newblock {\em Algorithms and Theory of Computation Handbook, Second Edition},
  chapter Self-stabilizing Algorithms, pages 26.1--26.45.
\newblock Chapman \& Hall/CRC Applied Algorithms and Data Structures. CRC
  Press, Taylor \& Francis Group, November 2009.

\end{thebibliography}

\end{document}